\documentclass[11pt,letterpaper]{article}
\usepackage[margin=1in]{geometry}
\usepackage[T1]{fontenc}
\usepackage[utf8]{inputenc}
\usepackage{newtxtext,amsmath,amsthm,mathtools,newtxmath}
\usepackage{microtype,booktabs,listings,textcomp}
\usepackage{graphicx}
\usepackage{algorithm,algpseudocode,needspace}
\algrenewcommand{\algorithmicrequire}{\textbf{Input:}}
\algrenewcommand{\algorithmicensure}{\textbf{Output:}}
\usepackage[numbers,sort&compress]{natbib}
\usepackage[hidelinks]{hyperref}
\newtheorem{theorem}{Theorem}
\newtheorem{lemma}{Lemma}
\theoremstyle{remark}
\newtheorem{remark}{Remark}
\theoremstyle{plain}
\providecommand{\doi}[1]{\href{https://doi.org/#1}{doi: \nolinkurl{#1}}}
\allowdisplaybreaks[2]
\hypersetup{pdftitle={Approximating Prize-Collecting TSP below 1.556}}
\title{Approximating Prize-Collecting TSP below 1.556}
\author{%
Hong Li\\
\small School of Mathematics and Statistics, Yunnan University\\
\small \texttt{honglimath@126.com}
}
\date{}
\begin{document}
\maketitle

\begin{abstract}
The prize-collecting traveling salesperson problem is a variant of the metric traveling salesperson problem in which vertices may be left unvisited by paying their associated penalties. The objective is to minimize the length of the tour plus the total penalty of the unvisited vertices. Blauth, Klein, and N\"agele gave the previously best-known LP-relative $1.599$-approximation. Their algorithm constructs tours by applying a splitting-off preprocessing to an LP solution, decomposing the resulting solution into trees, pruning them, and performing parity correction. We improve the LP-relative approximation ratio to $1.55576$. Our main contribution is a new bound on the cost of parity correction that exploits parity information revealed by further pruning in the analysis. This bound makes the splitting-off preprocessing unnecessary, so the improvement is achieved by a simpler version of their algorithm. The stated factor follows from a numerical parameter search with exact verification, while an elementary analysis already yields an LP-relative $1.56$-approximation.
\end{abstract}

\section{Introduction}
\subsection{Model and LP relaxation}
In the prize-collecting traveling salesperson problem (PCTSP), the input is a complete undirected graph $G=(V,E)$ with nonnegative edge lengths $c$ satisfying the triangle inequality, a root $r\in V$, and a nonnegative penalty $\pi_v$ for each $v\in V\setminus\{r\}$. A feasible solution is a tour containing $r$, and the objective is to minimize the length of the tour plus the total penalty of the vertices not visited by the tour. We allow the root-only tour and, for a tour visiting exactly two vertices, the same edge to be traversed twice.

We use $\pi_r:=0$ for convenience. For $S\subseteq V$, denote $\delta(S):=\{e\in E:|e\cap S|=1\}$; for $v\in V$, use $\delta(v):=\delta(\{v\})$. For a vector $z\in\mathbb R^E$ and an edge set $F\subseteq E$, write $z(F):=\sum_{e\in F}z_e$. The variable $x_e$ represents the fractional use of edge $e$, while $y_v$ represents the extent to which vertex $v$ is visited. The natural LP relaxation of PCTSP is
\begin{equation}\label{eq:lp}
\begin{aligned}
 \min\quad &\sum_{e\in E}c_ex_e+\sum_{v\in V}\pi_v(1-y_v)\\
 \text{subject to}\quad &x(\delta(v))=2y_v &&\forall v\in V\setminus\{r\},\\
 &x(\delta(r))\le2,\\
 &x(\delta(S))\ge2y_v &&\forall S\subseteq V\setminus\{r\},\ v\in S,\\
 &y_r=1,\\
 &x_e\ge0 &&\forall e\in E,\\
 &y_v\ge0 &&\forall v\in V.
\end{aligned}
\end{equation}
The constraints imply $y_v\le1$ for every $v\in V$. In the LP objective, $c^\top x=\sum_{e\in E}c_ex_e$ is the fractional edge-length term and $\sum_{v\in V}\pi_v(1-y_v)$ is the fractional penalty term. Let $(x^*,y^*)$ be an optimal solution of~\eqref{eq:lp}, and write $\mathrm{OPT}_{\mathrm{LP}}:=c^\top x^*+\sum_{v\in V}\pi_v(1-y_v^*)$ for its objective value. Every feasible PCTSP tour induces a feasible solution of~\eqref{eq:lp} with the same objective value, so $\mathrm{OPT}_{\mathrm{LP}}$ is a lower bound on the optimum PCTSP objective value. An LP-relative $\alpha$-approximation returns a tour with objective value at most $\alpha\mathrm{OPT}_{\mathrm{LP}}$. For each $v\ne r$, the cut constraints indexed by $v$ can be separated by a minimum $r$--$v$ cut; the remaining constraints are explicit. An optimal LP solution can be computed in polynomial time by the ellipsoid method~\cite{GLS1981}.

\subsection{Related work}
The metric traveling salesperson problem (TSP) is defined on a complete undirected graph with nonnegative edge lengths satisfying the triangle inequality and asks for a shortest tour visiting all vertices. Prize-collecting variants of TSP were first studied by Balas~\cite{Balas1989}. For the formulation of PCTSP considered here, Bienstock, Goemans, Simchi-Levi, and Williamson~\cite{BGW1993} gave an LP-relative $2.5$-approximation based on threshold rounding. Goemans and Williamson~\cite{GW1995} subsequently gave a $2$-approximation through a primal-dual approach. The factor of $2$ was first beaten by Archer, Bateni, Hajiaghayi, and Karloff~\cite{ABHK2011}, who obtained an approximation ratio of approximately $1.979$. Goemans~\cite{Goemans2009} then combined randomized threshold rounding with the primal-dual algorithm to obtain a $1.91457$-approximation. Blauth and N\"agele~\cite{BN2023} later obtained an LP-relative $1.774$-approximation based on a refined decomposition of the LP solution into trees.

Blauth, Klein, and N\"agele (BKN)~\cite{BKN2026} obtained the previously best-known LP-relative $1.599$-approximation. The BKN algorithm starts with an optimal solution $(x^*,y^*)$ of~\eqref{eq:lp}. For a splitting threshold $\delta\in[0,1)$, the splitting-off preprocessing performs a complete splitting at every vertex $v$ with $y_v^*<\delta$~\cite[Theorem~3]{BKN2026}, where a complete splitting at $v$~\cite{BN2023} is a sequence of splitting-off operations~\cite{Lovasz1976,Mader1978,Frank1992} that reduces $y_v$ to zero.\footnote{The construction of the tree decomposition in~\cite[Section~6]{BKN2026} also uses complete splittings; the splitting-off preprocessing refers only to the complete splittings performed before the decomposition.} The preprocessing produces a feasible solution $(x',y')$ with $y_v'=0$ whenever $y_v^*<\delta$, $y_v'=y_v^*$ otherwise, and $c^\top x'\le c^\top x^*$. The algorithm then applies a tree decomposition to $(x',y')$, obtaining a probability distribution over trees containing $r$ (see Lemma~\ref{lem:decomp}). For a tree $T$ in the decomposition and a pruning threshold $\gamma\in[0,1]$, the core of $T$ at pruning threshold $\gamma$ is the inclusion-wise minimal subtree of $T$ spanning $r$ and all vertices $v$ of $T$ with $y_v'\ge\gamma$; the algorithm prunes $T$ to this core. It then performs parity correction, as in the Christofides--Serdyukov algorithm for TSP~\cite{Christofides1976,Serdyukov1978}, by adding a minimum-length perfect matching on the odd-degree vertices of the core. The resulting connected multigraph has even degree at every vertex, and shortcutting an Euler tour gives a tour on the vertices of the core. The algorithm enumerates the candidate splitting thresholds $\delta$, the trees in the corresponding decompositions, and the candidate pruning thresholds $\gamma$, and returns the tour with the least objective value. The algorithm is deterministic; randomness is used only in the analysis. A first analysis, with a random pruning threshold $\gamma$ and a fixed splitting threshold $\delta$, yields the golden ratio $(1+\sqrt{5})/2$; a refined analysis that also randomizes $\delta$ improves the guarantee to $1.599$.

Several variants and special cases of PCTSP have also been studied. The prize-collecting stroll problem (PCS), the path version of PCTSP with two prescribed endpoints, has been considered in several works~\cite{ABHK2011,AKS2015,BKN2026,Li2026}. BKN~\cite{BKN2026} obtained an LP-relative $1.6662$-approximation for PCS. More recently, we~\cite{Li2026} showed that, for every fixed $\varepsilon>0$, a polynomial-time $\rho$-approximation for PCTSP yields a polynomial-time $(\rho+\varepsilon)$-approximation for PCS. Improvements in the approximation ratio for PCTSP transfer to PCS up to an arbitrarily small additive loss. Alimi, M\"omke, and Ruderer~\cite{AMR2025} studied prize-collecting ordered TSP and prize-collecting multi-path TSP. For PCTSP in special metric spaces, polynomial-time approximation schemes are known for planar graph metrics~\cite{Planar2011} and metrics of bounded doubling dimension~\cite{CJJ2020}.

\subsection{Our results}\label{sec:results}

Our main contribution is a new bound on the cost of parity correction that exploits parity information revealed by further pruning in the analysis; see Section~\ref{sec:why}. This bound makes the splitting-off preprocessing unnecessary, so we consider a simpler version of the BKN algorithm that omits this preprocessing: starting from an optimal LP solution $(x^*,y^*)$, it applies the tree decomposition directly to $(x^*,y^*)$ and then performs the same pruning and parity-correction steps as BKN. We state this algorithm as Algorithm~\ref{alg:bkn}. It is the BKN algorithm with the splitting threshold fixed to $\delta=0$, although no splitting threshold is needed in our formulation. Using our bound, an elementary analysis already shows that Algorithm~\ref{alg:bkn} is an LP-relative $1.56$-approximation; see Theorem~\ref{thm:simple}. A numerical parameter search with exact verification gives the following stronger guarantee.

\begin{theorem}\label{thm:main}
Algorithm~\ref{alg:bkn} runs in polynomial time and returns a tour satisfying
\begin{equation}\label{eq:main}
 \mathrm{ALG}\le 1.55576\,\mathrm{OPT}_{\mathrm{LP}}
 \le 1.55576\,\mathrm{OPT},
\end{equation}
where $\mathrm{OPT}$ denotes the optimum value of the PCTSP instance and $\mathrm{ALG}$ denotes the objective value of the tour returned by Algorithm~\ref{alg:bkn}.
\end{theorem}

The same approximation guarantee also holds for the BKN algorithm. In their derandomization, BKN enumerate the splitting thresholds in $\{y_v^*:v\in V\}$, obtained from the original optimal LP solution $(x^*,y^*)$~\cite[Section~4, proof of Theorem~1]{BKN2026}. Taking $\delta=\min_{v\in V}y_v^*$ causes no vertex to satisfy $y_v^*<\delta$, so the splitting-off preprocessing does nothing and this case is exactly the no-splitting case analyzed here. Although BKN state their algorithm for splitting thresholds $\delta\in[0,1)$~\cite[Algorithm~1]{BKN2026}, this minimum may equal $1$; in that case, all vertices satisfy $y_v^*=1$, and the preprocessing is again vacuous, exactly as for $\delta=0$. Since the BKN algorithm returns the best tour over the enumerated splitting thresholds, it inherits the same $1.55576$ guarantee. BKN reported computational evidence suggesting that an analysis following their proof could not achieve a ratio of $1.59$, while noting that this does not exclude a better guarantee for their algorithm~\cite[Remark~1]{BKN2026}. Our result shows that such a better guarantee indeed holds.

We also study how far the analysis can be improved by changing only the parameter distribution. The inequalities used to bound the expected tour length and expected penalty in Lemma~\ref{lem:scalar} cannot certify a factor below approximately $1.55562$, even when arbitrary joint distributions over the full parameter range are allowed. Our guarantee is within $1.4\cdot10^{-4}$ of the best factor obtainable from those inequalities; see Remark~\ref{rem:lower} and Appendix~\ref{app:lower}.

\subsection{Why the analysis improves}\label{sec:why}

\paragraph{Parity correction.} The objective value of a tour consists of its length and the penalties of the vertices it omits. Once a core is fixed, the penalty term is determined. Shortcutting an Euler tour of the multigraph formed by the core and the matching used for parity correction gives a tour whose length is at most the length of the core plus the length of that matching. The key issue is therefore how to bound the length of this matching after pruning. For a graph or multigraph $H$, write $V(H)$ and $E(H)$ for its vertex set and edge multiset, and let $\operatorname{odd}(H)$ be its set of odd-degree vertices. For an edge set or multiset $F$, write $c(F):=\sum_{e\in F}c_e$, counting multiplicities, and let $\chi^F$ be its incidence vector, also counting multiplicities. Write $c(H):=c(E(H))$.

For an even-cardinality set $Q\subseteq V$, a $Q$-join is an edge multiset whose odd-degree vertices are exactly $Q$. In our complete metric graph, the minimum $Q$-join length equals the length of a minimum-length perfect matching on $Q$. The algorithm computes a minimum-length perfect matching, while the analysis may bound its length using $Q$-joins. The $Q$-join dominant is
\begin{equation}\label{eq:join}
 P_Q^\uparrow=\left\{z\in\mathbb{R}_{\ge0}^E:z(\delta(S))\ge1\quad\forall S\subseteq V\text{ with }|S\cap Q|\text{ odd}\right\}.
\end{equation}
The integrality theorem of Edmonds and Johnson~\cite{EJ1973} implies that every vector $z\in P_Q^\uparrow$ gives an upper bound $c^\top z$ on the minimum $Q$-join length, and hence on the length of the matching used for parity correction. When $Q$ is the set of odd-degree vertices of the core, the analysis seeks a vector $z\in P_Q^\uparrow$ for which $c^\top z$ is small.

For a core $R$, set $Q=\operatorname{odd}(R)$. The BKN analysis constructs such a vector from the edge vector $x'$ of the LP solution after the splitting-off preprocessing and the edges of $R$~\cite[proof of Lemma~2]{BKN2026}. For an edge $e\in E(R)$, let $h'_R(e)$ be the maximum of $y_v'$ over the vertices in the connected component of $(V(R),E(R)\setminus\{e\})$ not containing $r$. The vector is
\begin{equation}\label{eq:bkn-vector}
 z^{\mathrm{BKN}}=\frac{x'}{3-\delta}+\sum_{e\in E(R)}\left(1-\frac{2h'_R(e)}{3-\delta}\right)\chi^{\{e\}}.
\end{equation}
The coefficient $1/(3-\delta)$ of $x'$ and the coefficients $1-2h'_R(e)/(3-\delta)$ on the edges of $R$ must together ensure the cut inequalities in~\eqref{eq:join}. Subject to these inequalities, smaller coefficients lead to a smaller value of $c^\top z^{\mathrm{BKN}}$ and can therefore improve the resulting LP-relative upper bound on the cost of parity correction.

\paragraph{Bottleneck of the BKN analysis.} To see both why the splitting-off preprocessing is useful in the BKN analysis and why it becomes a bottleneck for further improvement, fix a set $S$ with $r\notin S$ and $|S\cap Q|$ odd. Summing the degrees in $R$ of the vertices of $S$, as in the handshake lemma, gives $|\delta(S)\cap E(R)|\equiv |S\cap Q|\pmod 2$, so the number of edges of $R$ crossing the cut is odd. To prove that the vector in~\eqref{eq:bkn-vector} belongs to the $Q$-join dominant $P_Q^\uparrow$, BKN consider separately cuts crossed by exactly one edge of $R$ and cuts crossed by at least three edges of $R$. If $e$ is the unique crossing edge, then the cut constraints in~\eqref{eq:lp} give $x'(\delta(S))\ge2h'_R(e)$, and hence $z^{\mathrm{BKN}}(\delta(S))\ge2h'_R(e)/(3-\delta)+1-2h'_R(e)/(3-\delta)=1$. If at least three edges of $R$ cross the cut, the argument uses the lower bound $x'(\delta(S))\ge2\delta$ provided by the splitting-off preprocessing and the bound $h'_R(e)\le1$ for the crossing edges, yielding $z^{\mathrm{BKN}}(\delta(S))\ge2\delta/(3-\delta)+3(1-2/(3-\delta))=1$. Thus, taking a positive $\delta$ allows the coefficients $1-2h'_R(e)/(3-\delta)$ on the edges of $R$ to be smaller while the vector remains in $P_Q^\uparrow$, although the coefficient $1/(3-\delta)$ of $x'$ becomes larger. The decrease in the coefficients on the edges of $R$ can outweigh this increase and yield a smaller overall upper bound on the cost of parity correction.

The bottleneck for further improvement comes from the penalty term. By the vertex-marginal property of the tree decomposition (see Lemma~\ref{lem:decomp}), a vertex with $y_v'=0$ appears in no tree in the support of the decomposition. Every vertex with $y_v^*<\delta$, which is completely split by the preprocessing, is absent from every subsequent core and every resulting tour. Its full penalty is paid. For a fixed splitting threshold $\delta$, the full penalty $\pi_v$ is at most $1/(1-\delta)$ times its fractional penalty $\pi_v(1-y_v^*)$. To bound the penalties of these omitted vertices uniformly over all $y_v^*<\delta$ by $\rho$ times their fractional penalties, we require $\delta\le1-1/\rho$; for $\rho=1.599$, this gives $\delta\lesssim0.3746$. In their refined analysis, BKN randomize $\delta$ over an interval $[\kappa_0,\kappa]$ and choose $\kappa_0=0.3724$ for their $1.599$ guarantee~\cite[proof of Theorem~1]{BKN2026}. Vertices with $y_v^*<\kappa_0$ are completely split for every realization of $\delta$, giving the factor $1/(1-\kappa_0)$ appearing in the penalty bound~\cite[Equation~(8)]{BKN2026}. The value $\kappa_0=0.3724$ is already close to the upper limit $1-1/1.599\approx0.3746$ imposed by this penalty term. On the other hand, taking $\delta=0$ avoids the penalty loss caused by splitting, but makes the lower bound $x'(\delta(S))\ge2\delta$ vacuous for cuts crossed by at least three edges of $R$. With this analysis, the bound in~\cite[Theorem~4]{BKN2026} evaluates to $5/3$ at $\delta=0$.

\paragraph{Our mechanism.} We keep the original optimal LP solution $(x^*,y^*)$ and construct a different vector $z\in P_Q^\uparrow$; see~\eqref{eq:certificate}. For each cut $S$ with $r\notin S$ and $|S\cap Q|$ odd, we combine a cut-specific lower bound on $x^*(\delta(S))$, determined by the largest $y_v^*$ in $S$, with parity information revealed by further pruning the core $R$. This parity information is captured by auxiliary edges introduced below: they compensate for part of the loss in the cut inequality when the coefficients on the edges of $R$ are reduced, while a potential function bounds their expected length. Together, these yield a new upper bound on the cost of parity correction.

\paragraph{Auxiliary threshold.} We introduce an auxiliary threshold $t\in[0,1]$, analogous to the pruning threshold $\gamma$. For any tree $H$ containing $r$, let $\operatorname{core}(H,t)$ be the inclusion-wise minimal subtree of $H$ spanning $r$ and all vertices $v\in V(H)$ with $y_v^*\ge t$. Let $T$ be a tree in the decomposition of $(x^*,y^*)$, let $R=\operatorname{core}(T,\gamma)$ be the core used by the algorithm, and set $Q=\operatorname{odd}(R)$. Define $U_t:=\{v\in V:y_v^*\ge t\}$ and $R_t:=\operatorname{core}(R,t)=\operatorname{core}(T,\max\{\gamma,t\})$. Thus, $R_t$ is the part of the fixed core $R$ needed to connect the root to the vertices of $R$ with $y_v^*\ge t$; in particular, $R_t=R$ for $t\le\gamma$. Unlike the pruning threshold $\gamma$, the auxiliary threshold $t$ and the auxiliary edges introduced below are used only in the analysis of the fixed core $R$.

Fix a cut $S$ with $|S\cap Q|$ odd and, by complementing $S$ if necessary, assume $r\notin S$. Let $h(S):=\max_{v\in S}y_v^*$, $m(S):=|\delta(S)\cap E(R)|$, and $m_t(S):=|\delta(S)\cap E(R_t)|$. Since $|S\cap Q|$ is odd, $m(S)$ is odd. The cut constraints in~\eqref{eq:lp} give $x^*(\delta(S))\ge2h(S)$. For cuts with $m(S)\ge3$, the BKN analysis uses the uniform bound $2\delta$ and the fact that $m(S)\ge3$; we instead use the cut-specific bound $2h(S)$ and, for each auxiliary threshold $t$, the parity of $m_t(S)$, which counts the edges that remain across the cut after further pruning.

\paragraph{Auxiliary edges.} To use the parity of $m_t(S)$ in constructing $z$, for every vertex $v\in\operatorname{odd}(R_t)\setminus U_t$, choose a vertex $u\in U_t$ minimizing $c_{uv}$ and add $\{u,v\}$ as an auxiliary edge. Let $I_t(R_t)$ be the multiset of these auxiliary edges. When $t>h(S)$ and $m_t(S)$ is odd, at least one auxiliary edge crosses the cut. Indeed, $S\cap U_t=\varnothing$, and the handshake argument above gives $|S\cap\operatorname{odd}(R_t)|\equiv m_t(S)\pmod 2$; hence $S$ contains a vertex of $\operatorname{odd}(R_t)\setminus U_t$, whose auxiliary edge leaves $S$ because its other endpoint lies in $U_t$. Figure~\ref{fig:core-auxiliary} illustrates the core $R$, the subtree $R_t$, and an auxiliary edge.

\begin{figure}[!htbp]
\centering
\includegraphics[width=0.7\linewidth]{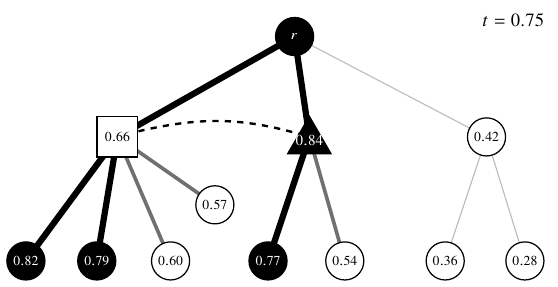}
\small
\caption[The core, auxiliary threshold, and auxiliary edge.]{The tree $T$, the core $R=\operatorname{core}(T,0.50)$, and $R_t=\operatorname{core}(R,t)$ for $t=0.75$. The numbers at the vertices are the values $y_v^*$, with $y_r^*=1$. Light gray edges belong to $E(T)\setminus E(R)$, medium gray edges to $E(R)\setminus E(R_t)$, and thick black edges to $E(R_t)$. Filled vertices belong to $U_t$. The square is a vertex $v_1\in\operatorname{odd}(R_t)\setminus U_t$, and the filled triangle is the chosen vertex $u\in U_t$ minimizing $c_{uv_1}$. The dashed curve is the auxiliary edge $\{u,v_1\}$.}
\label{fig:core-auxiliary}
\end{figure}

We use these auxiliary edges in constructing the vector $z$. In addition to terms involving $x^*$ and the edges of $R$, the vector contains a nonnegative combination of the incidence vectors $\chi^{I_t(R_t)}$ over a range of auxiliary thresholds $t$. Reducing the coefficients on the edges of $R$ decreases the cut value $z(\delta(S))$, and these incidence vectors offset part of this decrease: whenever $t>h(S)$ and $m_t(S)$ is odd, the crossing auxiliary edge guaranteed above adds value to the cut $S$. This is how the parity of $m_t(S)$ enters the cut inequalities in~\eqref{eq:join}. Combined with the cut-specific bound $x^*(\delta(S))\ge2h(S)$, this allows smaller coefficients on the edges of $R$ while keeping $z$ in $P_Q^\uparrow$ (Lemma~\ref{lem:certificate}).

The auxiliary edges may have positive length, so we must also bound their contribution to $c^\top z$. For fixed thresholds $\gamma$ and $t$, let $T$ be drawn from the probability distribution over trees given by the tree decomposition (see Lemma~\ref{lem:decomp}), and write $\mathbb{E}_T$ for the expectation over $T$. Using a potential function, we show in Lemma~\ref{lem:potential} that
\begin{equation}\label{eq:repair-overview}
 \mathbb{E}_T c(I_t(R_t))
 \le
 \mathbb{E}_T\bigl[c(T)-c(\operatorname{core}(T,t))\bigr],
\end{equation}
that is, the expected length of the auxiliary edges is bounded by the expected length of the edges removed from $T$ when forming $\operatorname{core}(T,t)$. Combining these two properties gives the new upper bound on the cost of parity correction (Lemma~\ref{lem:length}). It does not require a positive splitting threshold, so the splitting-off preprocessing can be omitted.

\paragraph{Overview.} Section~\ref{sec:bkn} describes the algorithm and the tree decomposition on which it is based. Section~\ref{sec:analysis} bounds the expected length of the auxiliary edges, constructs the vector used to bound the cost of parity correction, and proves the required cut inequalities. Section~\ref{sec:ratio} combines the resulting bound on the tour length with the expected penalty. Section~\ref{sec:simple-guarantee} proves the factor $1.56$ using elementary inequalities, and Section~\ref{sec:explicit-distribution} improves it to $1.55576$ using numerical parameter search and exact verification. Section~\ref{sec:conclusion} concludes the paper.

\section{The algorithm}\label{sec:bkn}

We now give a formal description of the simpler version of the BKN algorithm introduced in Section~\ref{sec:results}. For the description and analysis of the algorithm, we use the following notation. For a graph or multigraph $H$, let $\deg_H(v)$ be the degree of $v$, set to zero when $v\notin V(H)$. For distinct vertices $u,v$, write $c_{uv}:=c_{\{u,v\}}$, and set $c_{vv}:=0$.

\subsection{Tree decomposition}

The algorithm first computes an optimal solution $(x^*,y^*)$ of~\eqref{eq:lp}. We then apply the following tree decomposition to $(x^*,y^*)$ to obtain a family of trees; the decomposition was stated by Blauth and N\"agele~\cite[Lemma~12]{BN2023} and BKN give a proof of the form stated below~\cite[Lemma~1 and Section~6]{BKN2026}.

\begin{lemma}[Tree decomposition \cite{BN2023,BKN2026}]\label{lem:decomp}
Given a rational feasible solution $(x,y)$ of~\eqref{eq:lp}, one can compute in polynomial time a family $\mathcal T$ of polynomial size consisting of trees containing $r$, and rational probabilities $\mu_T$, such that
\begin{equation}\label{eq:decomp}
 \sum_{T\in\mathcal T}\mu_T=1,\qquad
 \sum_{T\in\mathcal T}\mu_T\chi^{E(T)}\le x,\qquad
 \sum_{T\in\mathcal T:\,v\in V(T)}\mu_T=y_v\quad\forall v\in V.
\end{equation}
\end{lemma}

For the analysis, let $\mu$ denote the distribution that assigns probability $\mu_T$ to each tree $T$ in the family $\mathcal T$ obtained from $(x^*,y^*)$, and write $\mathbb{E}_T$ for expectation over $T\sim\mu$. Equation~\eqref{eq:decomp} and the LP degree constraints give
\[
 \Pr[v\in V(T)]=y_v^*,\qquad
 \mathbb{E}_T c(T)\le c^\top x^*,\qquad
 \mathbb{E}_T\deg_T(v)\le2y_v^*\quad\forall v\ne r.
\]

\subsection{Cores and parity correction}

We apply the following pruning and parity-correction steps to every tree obtained from the tree decomposition. Recall that, for any tree $H$ containing $r$ and any pruning threshold $\gamma\in[0,1]$, $\operatorname{core}(H,\gamma)$ is obtained from $H$ by removing the edges that are not needed to connect $r$ to all vertices $v\in V(H)$ with $y_v^*\ge\gamma$; equivalently, it is the inclusion-wise minimal subtree with this property, following BKN~\cite[Definition~1]{BKN2026}.

For $T\in\mathcal T$, let $J_{T,\gamma}$ be a minimum-length perfect matching on $\operatorname{odd}(\operatorname{core}(T,\gamma))$; if this set is empty, take the empty matching. Form $H_{T,\gamma}$ on $V(\operatorname{core}(T,\gamma))$ with edge multiset $E(\operatorname{core}(T,\gamma))\uplus J_{T,\gamma}$, where $\uplus$ denotes multiset union.

The multigraph $H_{T,\gamma}$ is connected and has even degree at every vertex, so shortcutting an Euler tour gives a tour $\tau_{T,\gamma}$ visiting exactly $V(\operatorname{core}(T,\gamma))$. A root-only core gives the root-only tour, and a two-vertex core gives the tour traversing its edge twice.

Algorithm~\ref{alg:bkn} considers every $T\in\mathcal T$ and every pruning threshold $\gamma\in\Gamma:=\{y_v^*:v\in V\}$, and returns a tour with the least objective value.

\begin{algorithm}[H]
\caption{A simpler version of the BKN algorithm}\label{alg:bkn}
\begin{algorithmic}[1]
\Require A PCTSP instance $(G=(V,E),r,c,\pi)$.
\Ensure A tour containing $r$.
\State Compute an optimal solution $(x^*,y^*)$ of~\eqref{eq:lp}.
\State Compute a family $\mathcal T$ of trees through Lemma~\ref{lem:decomp} applied to $(x^*,y^*)$.
\State Set $\Gamma\gets\{y_v^*:v\in V\}$.
\ForAll{$T\in\mathcal T$ and $\gamma\in\Gamma$}
    \State Compute a minimum-length perfect matching $J_{T,\gamma}$ on $\operatorname{odd}(\operatorname{core}(T,\gamma))$.
    \State Form $H_{T,\gamma}$ with edge multiset $E(\operatorname{core}(T,\gamma))\uplus J_{T,\gamma}$.
    \State Obtain $\tau_{T,\gamma}$ by shortcutting an Euler tour of $H_{T,\gamma}$.
\EndFor
\State \Return a tour $\tau_{T,\gamma}$ of minimum objective value over all $T\in\mathcal T$ and $\gamma\in\Gamma$.
\end{algorithmic}
\end{algorithm}

Algorithm~\ref{alg:bkn} runs in polynomial time. The family $\mathcal T$ has polynomial size and $|\Gamma|\le|V|$, so at most $|\mathcal T||V|$ tours are computed. The LP solution, tree decomposition, each core, and each minimum-length perfect matching~\cite[Sections~4 and~7]{Edmonds1965} can be computed in polynomial time, as can each Euler tour and its shortcutting. For a comparison of the running time with BKN, see Remark~\ref{rem:shared-decomposition}.

As in the BKN analysis~\cite[Sections~2 and~3]{BKN2026}, we can bound $\mathrm{ALG}$ by analyzing a random choice of $T$ and $\gamma$. The thresholds in $\Gamma$ cover all cores obtained for $\gamma\in[0,1]$: increasing $\gamma$ to the smallest value in $\Gamma$ at least $\gamma$ does not change which vertices satisfy $y_v^*\ge\gamma$. This value exists because $1=y_r^*\in\Gamma$. Since shortcutting does not increase length and preserves the vertices of the core, for every $T\in\mathcal T$ and $\gamma\in[0,1]$ we have
\begin{equation}\label{eq:candidate-objective}
 \mathrm{ALG}\le c(\operatorname{core}(T,\gamma))+c(J_{T,\gamma})+\sum_{v\notin V(\operatorname{core}(T,\gamma))}\pi_v.
\end{equation}
Taking expectation of the right-hand side under any distribution on $\mathcal T\times[0,1]$ gives an upper bound on $\mathrm{ALG}$. We take $T\sim\mu$, and in Section~\ref{sec:ratio} we choose a distribution for the pruning threshold $\gamma$, independent of $T$, to bound this expectation.

\section{A stronger bound on parity correction}\label{sec:analysis}

As explained in Section~\ref{sec:why}, every vector $z\in P_Q^\uparrow$ with $Q=\operatorname{odd}(\operatorname{core}(T,\gamma))$ gives $c(J_{T,\gamma})\le c^\top z$. We first bound the expected length of the auxiliary edges, then use their incidence vectors to construct such a vector and bound the length of $\operatorname{core}(T,\gamma)$ together with its parity-correction matching.

\subsection{Bounding the length of the auxiliary edges}

For an auxiliary threshold $t\in[0,1]$, recall from Section~\ref{sec:why} that $U_t=\{v\in V:y_v^*\ge t\}$. For each $v\in V$, define $f_t(v):=\min_{u\in U_t}c_{uv}$, the distance from $v$ to $U_t$; it is well defined because $r\in U_t$. We have $f_t(v)=0$ for every $v\in U_t$, in particular $f_t(r)=0$, and the triangle inequality gives $|f_t(u)-f_t(v)|\le c_{uv}$. The function $f_t$ depends only on $y^*$ and the edge lengths $c$; in particular, it does not depend on the tree $T$.

For a core $R$, recall also that $R_t=\operatorname{core}(R,t)$ and that $I_t(R_t)$ contains one auxiliary edge for each vertex $v\in\operatorname{odd}(R_t)\setminus U_t$, joining $v$ to a vertex $u\in U_t$ minimizing $c_{uv}$, with ties broken by a fixed order on $V$. This auxiliary edge therefore has length $f_t(v)$. With this rule, $I_t(R_t)$ depends only on $U_t$ and $R_t$.

\paragraph{Potential function.} To bound these auxiliary edge lengths, let $H$ be a tree containing $r$. For each edge $e\in E(H)$, let $u_e$ and $w_e$ denote its endpoints closer to and farther from $r$ in $H$, respectively; that is, $u_e$ lies on the path in $H$ from $r$ to $w_e$. Define the potential function
\begin{equation}\label{eq:potential}
 \Phi_t(H):=\sum_{v\in V(H)\setminus\{r\}} f_t(v)(\deg_H(v)-2)
 =\sum_{e\in E(H)}\bigl(f_t(u_e)-f_t(w_e)\bigr).
\end{equation}

We call the two sums in~\eqref{eq:potential} the vertex expression and the edge expression, respectively. To see the equality, fix a nonroot vertex $v$. Exactly one edge $e$ at $v$ has $w_e=v$, namely the first edge on the path from $v$ to $r$, and it contributes $-f_t(v)$ to the edge expression; each of the other $\deg_H(v)-1$ edges $e$ at $v$ has $u_e=v$ and contributes $+f_t(v)$. Hence the total coefficient of $f_t(v)$ is $\deg_H(v)-2$, and the terms involving $r$ vanish because $f_t(r)=0$.

The vertex expression can be used to bound the length of the auxiliary edges, while the edge expression relates changes in $\Phi_t(H)$ to the length of the edges removed by pruning. The following lemma combines the two expressions to prove the bound~\eqref{eq:repair-overview} stated in Section~\ref{sec:why}.

\begin{lemma}\label{lem:potential}
Fix a pruning threshold $\gamma\in[0,1]$ and an auxiliary threshold $t\in[0,1]$. For each $T\in\mathcal T$, let $R=\operatorname{core}(T,\gamma)$ and $R_t=\operatorname{core}(R,t)=\operatorname{core}(T,\max\{\gamma,t\})$; both depend on $T$. Then, for $T\sim\mu$,
\[
 \mathbb{E}_T c(I_t(R_t))
 \le
 \mathbb{E}_T\bigl[c(T)-c(\operatorname{core}(T,t))\bigr].
\]
\end{lemma}

\begin{proof}
Fix $T\in\mathcal T$. We first prove the chain of inequalities
\[
 c(I_t(R_t))
 \le\Phi_t(R_t)
 \le\Phi_t(\operatorname{core}(T,t))
 \le\Phi_t(T)+c(T)-c(\operatorname{core}(T,t))
\]
and then take expectations over $T\sim\mu$.

Pruning first at $\gamma$ and then at $t$ is equivalent to pruning at the larger threshold, so $R_t=\operatorname{core}(T,\max\{\gamma,t\})$; in particular, $R_t=R$ for $t\le\gamma$, while larger values of $t$ may further prune $R$. Hence $R_t$ and $\operatorname{core}(T,t)$ are both cores of $T$ at thresholds at least $t$, and by minimality every nonroot leaf of each of them lies in $U_t$, where $f_t$ vanishes. Every other nonroot vertex of these trees has degree at least two, so all terms in the vertex expression of~\eqref{eq:potential} are nonnegative for both trees.

Let $v\in\operatorname{odd}(R_t)\setminus U_t$. Then $v$ is not a leaf of $R_t$, so its odd degree is at least three, and its term $f_t(v)(\deg_{R_t}(v)-2)$ is at least $f_t(v)$, the length of its auxiliary edge. Since all other terms are nonnegative, $c(I_t(R_t))\le\Phi_t(R_t)$.

Since $\max\{\gamma,t\}\ge t$, the tree $R_t$ is a subtree of $\operatorname{core}(T,t)$. Passing from $\operatorname{core}(T,t)$ to $R_t$ deletes vertices, whose terms in the vertex expression are nonnegative, and can only decrease the degrees of the remaining vertices. Since $f_t\ge0$, we obtain $\Phi_t(R_t)\le\Phi_t(\operatorname{core}(T,t))$.

Since $\operatorname{core}(T,t)$ is a subtree of $T$ containing $r$, the path from $r$ to each of its vertices is the same in both trees. Hence, for every edge $e$ of $\operatorname{core}(T,t)$, the endpoints $u_e$ and $w_e$ are the same in both trees, so $e$ contributes equally to the edge expressions of~\eqref{eq:potential} for $\Phi_t(\operatorname{core}(T,t))$ and $\Phi_t(T)$. Every other edge $e$ of $T$ contributes $f_t(u_e)-f_t(w_e)\ge-c_e$ to $\Phi_t(T)$. Summing over the edges of $T$ not in $\operatorname{core}(T,t)$ gives $\Phi_t(\operatorname{core}(T,t))\le\Phi_t(T)+c(T)-c(\operatorname{core}(T,t))$.

Finally, as noted after Lemma~\ref{lem:decomp}, every nonroot vertex $v$ satisfies $\Pr[v\in V(T)]=y_v^*$ and $\mathbb{E}_T\deg_T(v)\le2y_v^*$. By the vertex expression of~\eqref{eq:potential} and $f_t\ge0$,
\[
 \mathbb{E}_T\Phi_t(T)
 =\sum_{v\ne r}f_t(v)
 \bigl(\mathbb{E}_T\deg_T(v)-2\Pr[v\in V(T)]\bigr)
 \le0.
\]
Taking expectations over $T\sim\mu$ in the chain of inequalities above and using this bound, we obtain
\[
 \mathbb{E}_T c(I_t(R_t))
 \le\mathbb{E}_T\Phi_t(T)+\mathbb{E}_T\bigl[c(T)-c(\operatorname{core}(T,t))\bigr]
 \le\mathbb{E}_T\bigl[c(T)-c(\operatorname{core}(T,t))\bigr].
 \qedhere
\]
\end{proof}

\subsection{A feasible vector in the \texorpdfstring{$Q$}{Q}-join dominant}

For a fixed core $R$, we now construct a vector $z_R\in P_{\operatorname{odd}(R)}^\uparrow$ from $x^*$, the edges of $R$, and nonnegative combinations of the incidence vectors $\chi^{I_t(R_t)}$. The construction uses the auxiliary thresholds in an interval $[0,s]$, where $s\in[0,1]$ is a parameter; thus $s$ is the largest auxiliary threshold used. When we later take $R=\operatorname{core}(T,\gamma)$ and bound $c^\top z_R$ in expectation over $T$, Lemma~\ref{lem:potential} controls the contribution of these incidence vectors to $c^\top z_R$.

The coefficients on the edges of $R$ use the analogue of the quantity $h'_R(e)$ in~\eqref{eq:bkn-vector}. Since we omit the splitting-off preprocessing, we define it with $y^*$ in place of $y'$, and for every tree containing $r$ rather than only for cores. For a tree $H$ containing $r$ and an edge $e\in E(H)$, let $h_H(e)$ be the maximum of $y_v^*$ over the vertices in the connected component of $(V(H),E(H)\setminus\{e\})$ not containing $r$. Thus $h_H(e)$ is the largest pruning threshold $\gamma$ for which $e\in E(\operatorname{core}(H,\gamma))$. More precisely, for every $\gamma\in[0,1]$,
\begin{equation}\label{eq:height}
 e\in E(\operatorname{core}(H,\gamma))\iff h_H(e)\ge\gamma.
\end{equation}
This holds because $e$ is needed to connect $r$ to the vertices $v\in V(H)$ with $y_v^*\ge\gamma$ exactly when such a vertex lies in the connected component of $(V(H),E(H)\setminus\{e\})$ not containing $r$. Moreover, pruning does not change this quantity on the edges that remain:
\begin{equation}\label{eq:height-core}
 h_{\operatorname{core}(H,\gamma)}(e)=h_H(e)\qquad\forall e\in E(\operatorname{core}(H,\gamma)).
\end{equation}
Indeed, let $e\in E(\operatorname{core}(H,\gamma))$, and let $v$ be a vertex attaining the maximum in the definition of $h_H(e)$. By~\eqref{eq:height}, $y_v^*=h_H(e)\ge\gamma$, so $v$ belongs to $\operatorname{core}(H,\gamma)$. Since $e$ lies on the path from $r$ to $v$ in $H$, it also lies on this path in $\operatorname{core}(H,\gamma)$, so $h_{\operatorname{core}(H,\gamma)}(e)\ge y_v^*=h_H(e)$. The reverse inequality holds because pruning only removes vertices.

\needspace{9\baselineskip}
\begin{lemma}\label{lem:certificate}
Let $R$ be a fixed core, let $0\le s\le a\le1$ with $a>0$, and put $\lambda=1/(3a-s)$. For $t\in[0,1]$, let $R_t=\operatorname{core}(R,t)$. Define
\begin{equation}\label{eq:certificate}
 z_R=
 \underbrace{\lambda x^*}_{\text{$x^*$ term}}
 +
 \underbrace{\sum_{e\in E(R)}\bigl(1-2\lambda\min\{h_R(e),a\}\bigr)\chi^{\{e\}}}_{\text{$R$ term}}
 +
 \underbrace{2\lambda\int_0^s\chi^{I_t(R_t)}\,dt}_{\text{auxiliary term}}.
\end{equation}
Then $z_R\in P_{\operatorname{odd}(R)}^\uparrow$.
\end{lemma}

The vector $z_R$ modifies $\chi^{E(R)}$, which already lies in $P_{\operatorname{odd}(R)}^\uparrow$: it adds the $x^*$ term and the auxiliary term, and the $R$ term lowers the coefficient of each edge $e\in E(R)$ by $2\lambda\min\{h_R(e),a\}$. The auxiliary term is the nonnegative combination of the incidence vectors $\chi^{I_t(R_t)}$ mentioned above; it has only finitely many distinct terms, since $U_t$ and $R_t$, and hence $I_t(R_t)$, change only when $t$ passes a value $y_v^*$. Recall that $s$ is the largest auxiliary threshold used; the parameter $a$ caps the reduction on each edge of $R$. All three changes are proportional to $\lambda$, which depends only on $s$ and $a$, so the construction has only these two parameters. When $s=0$ and $a=1$, the vector $z_R$ reduces to the BKN vector in~\eqref{eq:bkn-vector} with $\delta=0$.

\begin{proof}
Since $s\le a$, we have $\lambda a\le1/2$, so the coefficients in the $R$ term are nonnegative. The $x^*$ term and the auxiliary term are also nonnegative. Hence $z_R\in\mathbb{R}_{\ge0}^E$, and it remains to verify the cut inequalities in~\eqref{eq:join}, namely,
\[
 z_R(\delta(S))\ge1
 \qquad
 \forall S\subseteq V\text{ with }|S\cap\operatorname{odd}(R)|\text{ odd}.
\]

Fix such a set $S$ and, by complementing $S$ if necessary, assume $r\notin S$. As in Section~\ref{sec:why}, let $h(S)=\max_{v\in S}y_v^*$, $m(S)=|\delta(S)\cap E(R)|$, and $m_t(S)=|\delta(S)\cap E(R_t)|$. Since $m(S)\equiv|S\cap\operatorname{odd}(R)|\pmod2$, the number $m(S)$ is odd, so either $m(S)=1$ or $m(S)\ge3$. The cut constraints in~\eqref{eq:lp}, applied to $S$, give $x^*(\delta(S))\ge2y_v^*$ for every $v\in S$, so in both cases the $x^*$ term contributes at least $2\lambda h(S)$ to $z_R(\delta(S))$.

Suppose first that $m(S)=1$, and let $e$ be the unique edge of $R$ crossing the cut. Since no other edge of $R$ crosses the cut and $r\notin S$, the connected component of $(V(R),E(R)\setminus\{e\})$ not containing $r$ lies entirely in $S$, so $h(S)\ge h_R(e)$, and the $x^*$ term contributes at least $2\lambda h_R(e)$. In the $R$ term, only $e$ crosses the cut, and it contributes $1-2\lambda\min\{h_R(e),a\}$; the auxiliary term contributes a nonnegative amount. Hence
\[
 z_R(\delta(S))
 \ge 2\lambda h_R(e)+1-2\lambda\min\{h_R(e),a\}
 \ge1.
\]

Now suppose that $m(S)\ge3$. The $R$ term contributes $\sum_{e\in\delta(S)\cap E(R)}\bigl(1-2\lambda\min\{h_R(e),a\}\bigr)$ to $z_R(\delta(S))$. To compare this contribution with the auxiliary term, which is an integral over $t$, we also write it as an integral over $t$. By~\eqref{eq:height}, a crossing edge $e\in\delta(S)\cap E(R)$ belongs to $R_t$ exactly for $t\le h_R(e)$, so within $[0,a]$ it is counted in $m_t(S)$ on an interval of length $\min\{h_R(e),a\}$. Hence
\[
 \sum_{e\in\delta(S)\cap E(R)}\min\{h_R(e),a\}
 =\int_0^a m_t(S)\,dt.
\]
Since $\sum_{e\in\delta(S)\cap E(R)}1=m(S)$, the $R$ term contributes $m(S)-2\lambda\int_0^a m_t(S)\,dt$ to $z_R(\delta(S))$.

For $h(S)<t\le s$, we have $S\cap U_t=\varnothing$. If $m_t(S)$ is even, then $m_t(S)\le m(S)-1$, since $m(S)$ is odd. If $m_t(S)$ is odd, then so is $|S\cap\operatorname{odd}(R_t)|$, so $S$ contains a vertex $v\in\operatorname{odd}(R_t)\setminus U_t$; as shown in Section~\ref{sec:why}, the auxiliary edge added for $v$ crosses the cut. In either case,
\begin{equation}\label{eq:parity-step}
 m_t(S)-\chi^{I_t(R_t)}(\delta(S))\le m(S)-1.
\end{equation}
For every $t\in[0,a]$, we have $m_t(S)\le m(S)$ and $\chi^{I_t(R_t)}(\delta(S))\ge0$. Since $s\le a$, we can split $[0,a]$ at $\min\{h(S),s\}$ and at $s$; on the middle part, every $t$ satisfies $h(S)<t\le s$, so~\eqref{eq:parity-step} applies. This gives
\begin{align*}
 &\int_0^a m_t(S)\,dt-\int_0^s\chi^{I_t(R_t)}(\delta(S))\,dt\\
 ={}&\int_0^{\min\{h(S),s\}}\bigl(m_t(S)-\chi^{I_t(R_t)}(\delta(S))\bigr)\,dt
   +\int_{\min\{h(S),s\}}^{s}\bigl(m_t(S)-\chi^{I_t(R_t)}(\delta(S))\bigr)\,dt\\
 &\quad+\int_s^a m_t(S)\,dt\\
 \le{}&m(S)\min\{h(S),s\}+\bigl(m(S)-1\bigr)\max\{s-h(S),0\}+m(S)(a-s)\\
 ={}&a\,m(S)-\max\{s-h(S),0\}.
\end{align*}
Adding the contributions of the three terms of~\eqref{eq:certificate} to $z_R(\delta(S))$ yields
\begin{align*}
 z_R(\delta(S))
 &\ge 2\lambda h(S)+m(S)-2\lambda\bigl(a\,m(S)-\max\{s-h(S),0\}\bigr)\\
 &=m(S)(1-2\lambda a)+2\lambda\max\{h(S),s\}\\
 &\ge3(1-2\lambda a)+2\lambda s=1.
\end{align*}
The last inequality uses $m(S)\ge3$, $1-2\lambda a\ge0$, and $\max\{h(S),s\}\ge s$. The choice $\lambda=1/(3a-s)$ makes this lower bound equal to $1$ in the critical case $m(S)=3$ and $h(S)\le s$. Thus $z_R$ satisfies every cut inequality in~\eqref{eq:join}.
\end{proof}

Figure~\ref{fig:hard-cut} illustrates the parity argument in the proof of Lemma~\ref{lem:certificate}, in particular inequality~\eqref{eq:parity-step}, on the tree of Figure~\ref{fig:core-auxiliary}.

\begin{figure}[H]
\centering
\includegraphics[width=0.7\linewidth]{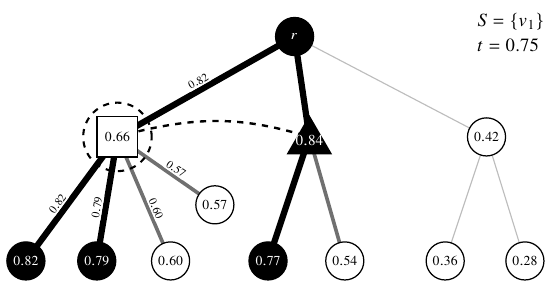}
\small
\caption[The cut $S=\{v_1\}$ and the auxiliary edge crossing it.]{Inequality~\eqref{eq:parity-step} for $R=\operatorname{core}(T,0.50)$ and the cut $S=\{v_1\}$, marked by the dashed oval. Vertex symbols and edge styles are as in Figure~\ref{fig:core-auxiliary}, with $R_t$ drawn for $t=0.75$: the square is $v_1$, the filled triangle is the chosen vertex $u\in U_t$, and the dashed curve is the auxiliary edge $\{u,v_1\}$. All five edges of $R$ at $v_1$ cross the cut, so $m(S)=5$. The numbers on these edges are their values $h_R(e)$; by~\eqref{eq:height}, such an edge belongs to $R_t$ exactly when $t\le h_R(e)$, and these values therefore determine $m_t(S)$. Since $h(S)=y_{v_1}^*=0.66$, inequality~\eqref{eq:parity-step} is needed for $t>0.66$. For $t\in(0.66,0.79]$, the three edges with $h_R(e)\ge0.79$ belong to $R_t$, so $m_t(S)=3$ is odd; throughout this interval, the auxiliary edge added for $v_1$ is the edge $\{u,v_1\}$ drawn for $t=0.75$, and it crosses the cut. Hence $m_t(S)-\chi^{I_t(R_t)}(\delta(S))\le2\le m(S)-1$. For $t\in(0.79,0.82]$, only the two edges with $h_R(e)=0.82$ remain, and for $t>0.82$ none does; thus $m_t(S)$ is even and hence at most $m(S)-1$, so~\eqref{eq:parity-step} holds as well.}
\label{fig:hard-cut}
\end{figure}

Applying Lemma~\ref{lem:certificate} to $R=\operatorname{core}(T,\gamma)$ and using Lemma~\ref{lem:potential} gives the following expected length bound. We write $\mathbf{1}_A$ for the indicator of a condition $A$.

\begin{lemma}\label{lem:length}
Fix a pruning threshold $\gamma\in[0,1]$ and parameters $0\le s\le a\le1$ with $a>0$, and put $\lambda=1/(3a-s)$. For $T\sim\mu$, the expected length of $\operatorname{core}(T,\gamma)$ together with its parity-correction matching satisfies
\begin{equation}\label{eq:length}
\begin{aligned}
 &\mathbb{E}_T\bigl[c(\operatorname{core}(T,\gamma))+c(J_{T,\gamma})\bigr]\\
 \le{}&\lambda c^\top x^*+\mathbb{E}_T\sum_{e\in E(T)}c_e
 \left[\bigl(2-2\lambda\min\{h_T(e),a\}\bigr)\mathbf{1}_{\{\gamma\le h_T(e)\}}
       +2\lambda\max\{s-h_T(e),0\}\right].
\end{aligned}
\end{equation}
\end{lemma}

\begin{proof}
For each $T$, let $R=\operatorname{core}(T,\gamma)$ and apply Lemma~\ref{lem:certificate}. Since $z_R\in P_{\operatorname{odd}(R)}^\uparrow$, the length of the parity-correction matching satisfies $c(J_{T,\gamma})\le c^\top z_R$. By~\eqref{eq:certificate} and $c^\top\chi^{I_t(R_t)}=c(I_t(R_t))$,
\[
 c(R)+c^\top z_R
 =\lambda c^\top x^*
 +\sum_{e\in E(R)}c_e\bigl(2-2\lambda\min\{h_R(e),a\}\bigr)
 +2\lambda\int_0^s c(I_t(R_t))\,dt,
\]
where the middle sum combines $c(R)$ with the $R$ term. We now rewrite the middle sum and bound the last term in expectation.

By~\eqref{eq:height} and~\eqref{eq:height-core}, an edge $e\in E(T)$ belongs to $R$ exactly when $\gamma\le h_T(e)$, and then $h_R(e)=h_T(e)$. Hence
\[
\sum_{e\in E(R)}c_e\bigl(2-2\lambda\min\{h_R(e),a\}\bigr)
= \sum_{e\in E(T)}c_e
 \bigl(2-2\lambda\min\{h_T(e),a\}\bigr)
 \mathbf{1}_{\{\gamma\le h_T(e)\}}.
\]

For each $t\in[0,s]$, Lemma~\ref{lem:potential} bounds $\mathbb{E}_T c(I_t(R_t))$, where $R_t=\operatorname{core}(R,t)$ as in that lemma. Integrating this bound over $t\in[0,s]$ gives
\begin{align*}
 \mathbb{E}_T\int_0^s c(I_t(R_t))\,dt
 &\le\int_0^s\mathbb{E}_T\bigl[c(T)-c(\operatorname{core}(T,t))\bigr]\,dt\\
 &=\mathbb{E}_T\sum_{e\in E(T)}c_e\max\{s-h_T(e),0\}.
\end{align*}
The equality follows from~\eqref{eq:height}: an edge $e\in E(T)$ is absent from $\operatorname{core}(T,t)$ for a set of thresholds in $[0,s]$ of length $\max\{s-h_T(e),0\}$.

Taking expectations in the identity for $c(R)+c^\top z_R$ and using the above displays, we obtain~\eqref{eq:length}.
\end{proof}

Compared with the case $s=0$, in which no auxiliary edges are used, introducing the auxiliary edges adds the term $2\lambda\,\mathbb{E}_T\sum_{e\in E(T)}c_e\max\{s-h_T(e),0\}$ to~\eqref{eq:length} and raises $\lambda$ from $1/(3a)$ to $1/(3a-s)$. The larger $\lambda$ increases $\lambda c^\top x^*$ but decreases $\mathbb{E}_T\sum_{e\in E(T)}c_e\bigl(2-2\lambda\min\{h_T(e),a\}\bigr)\mathbf{1}_{\{\gamma\le h_T(e)\}}$; this decrease is where the parity information described in Section~\ref{sec:why} enters the bound.

\section{The approximation guarantee}\label{sec:ratio}
\subsection{Reducing the guarantee to the parameter distribution}

Lemma~\ref{lem:length} holds for every fixed choice of $(\gamma,s,a)$. We now choose $(\gamma,s,a)$ at random, independently of $T\sim\mu$, from a joint distribution $\eta$ with $0\le\gamma\le1$ and $0\le s\le a\le1$, $a>0$; the three parameters may be correlated with one another. By~\eqref{eq:candidate-objective}, $\mathrm{ALG}$ is at most the expected length plus the expected penalty. We show that both are controlled by two functions of one variable that depend only on $\eta$, so the approximation guarantee is determined by $\eta$ alone.

For fixed $(\gamma,s,a)$ with $\lambda=1/(3a-s)$, let
\begin{equation}\label{eq:kernel}
 \ell_{\gamma,s,a}(y):=\lambda+\bigl(2-2\lambda\min\{y,a\}\bigr)\mathbf{1}_{\{\gamma\le y\}}+2\lambda\max\{s-y,0\}\qquad\forall y\in[0,1].
\end{equation}
This function is derived from~\eqref{eq:length}. Define
\[
 L(y):=\mathbb{E}_\eta\bigl[\ell_{\gamma,s,a}(y)\bigr],\qquad
 F(y):=\Pr_\eta[\gamma\le y]\qquad\forall y\in[0,1].
\]
Rather than normalizing the length bound in Lemma~\ref{lem:length} by $c^\top x^*$ for each fixed $(\gamma,s,a)$, we average $\ell_{\gamma,s,a}$ over $\eta$ before taking the worst case over $y$; this is possible because $\eta$ is independent of $T$. Thus $L$ is used to bound the expected length by a multiple of $c^\top x^*$ that does not depend on the instance. The distribution function $F$ of the pruning threshold is used to bound the expected penalty: a nonroot vertex with $y_v^*=y$ belongs to $T$ with probability $y$ and is retained in $\operatorname{core}(T,\gamma)$ whenever $\gamma\le y$, so its expected penalty is at most $\pi_v(1-yF(y))$.

\begin{lemma}\label{lem:scalar}
Suppose a constant $\rho$ satisfies
\begin{equation}\label{eq:scalar-conditions}
\begin{aligned}
 L(y)&\le\rho &&\forall y\in[0,1],\\
 1-yF(y)&\le\rho(1-y) &&\forall y\in[0,1].
\end{aligned}
\end{equation}
Then Algorithm~\ref{alg:bkn} satisfies $\mathrm{ALG}\le\rho\,\mathrm{OPT}_{\mathrm{LP}}$.
\end{lemma}

\begin{proof}
By~\eqref{eq:kernel}, the coefficient of $c_e$ in~\eqref{eq:length} equals $\ell_{\gamma,s,a}(h_T(e))-\lambda$, and $\ell_{\gamma,s,a}(y)\ge\lambda$ because $s\le a$ gives $\lambda a\le1/2$. Hence the first inequality in~\eqref{eq:scalar-conditions} gives $\mathbb{E}_\eta[\lambda]\le L(y)\le\rho$ for every $y$. Since $T$ and $(\gamma,s,a)$ are independent, taking expectation over $\eta$ in~\eqref{eq:length} bounds the expected length by
\begin{align*}
 \mathbb{E}_\eta[\lambda]\,c^\top x^*+\mathbb{E}_T\sum_{e\in E(T)}c_e\bigl(L(h_T(e))-\mathbb{E}_\eta[\lambda]\bigr)
 &\le\mathbb{E}_\eta[\lambda]\,c^\top x^*+\bigl(\rho-\mathbb{E}_\eta[\lambda]\bigr)\mathbb{E}_T c(T)\\
 &\le\rho\,c^\top x^*.
\end{align*}
The first step uses the first inequality in~\eqref{eq:scalar-conditions}, and the second uses $\mathbb{E}_T c(T)\le c^\top x^*$ and $\mathbb{E}_\eta[\lambda]\le\rho$. By the second inequality in~\eqref{eq:scalar-conditions}, the expected penalty is at most
\[
 \sum_{v\ne r}\pi_v\bigl(1-y_v^*F(y_v^*)\bigr)
 \le\rho\sum_{v\ne r}\pi_v(1-y_v^*).
\]
Taking expectation of~\eqref{eq:candidate-objective} and adding these two bounds gives $\mathrm{ALG}\le\rho\,\mathrm{OPT}_{\mathrm{LP}}$, since $y_r^*=1$.
\end{proof}

The next two subsections construct distributions $\eta$ satisfying~\eqref{eq:scalar-conditions}.

\subsection{A \texorpdfstring{$1.56$}{1.56}-approximation guarantee with fixed \texorpdfstring{$s$}{s} and \texorpdfstring{$a$}{a}}\label{sec:simple-guarantee}

We first apply Lemma~\ref{lem:scalar} with fixed values of $s,a$ and a random pruning threshold $\gamma$. The factor obtained in this setting can be optimized further; for simplicity we use the following rational parameters, for which the $1.56$ guarantee follows from elementary inequalities.

\begin{theorem}\label{thm:simple}
Algorithm~\ref{alg:bkn} runs in polynomial time and returns a tour satisfying
\[
 \mathrm{ALG}\le 1.56\,\mathrm{OPT}_{\mathrm{LP}}
 \le 1.56\,\mathrm{OPT}.
\]
\end{theorem}

\begin{proof}
Fix
\[
 s=\frac{27}{35},\qquad a=\frac{34}{35},\qquad
 \lambda=\frac{1}{3a-s}=\frac7{15}.
\]
Let $\eta$ be the joint distribution of $(\gamma,s,a)$, where $\gamma$ is chosen independently of $T$ with distribution function
\[
 F(y)=
 \begin{cases}
  0,&0\le y<14/39,\\
  \displaystyle\frac{39}{25}-\frac{14}{25y},&14/39\le y\le1.
 \end{cases}
\]
Then $1-yF(y)\le(39/25)(1-y)$ for every $y\in[0,1]$, with equality for $y\ge14/39$, so the second inequality in~\eqref{eq:scalar-conditions} holds with $\rho=39/25$.

Since $s,a$ are fixed, taking expectation in~\eqref{eq:kernel} gives
\[
 L(y)=\lambda+\bigl(2-2\lambda\min\{y,a\}\bigr)F(y)+2\lambda\max\{s-y,0\}.
\]
For $y\le14/39$, we have $F(y)=0$, so $L(y)\le\lambda(1+2s)=89/75<39/25$. For $y\ge a$, we have $F(y)\le1$, so $L(y)\le\lambda+2-2\lambda a=39/25$. For the remaining $y$, expansion and the arithmetic--geometric mean inequality give
\[
\begin{aligned}
 L(y)&=\frac{1811-896y-420/y}{375}\le\frac{1811-224\sqrt{30}}{375}&&\text{if }14/39\le y\le s,\\
 L(y)&=\frac{1541-546y-420/y}{375}\le\frac{1541-84\sqrt{130}}{375}&&\text{if }s\le y\le a.
\end{aligned}
\]
Both bounds are less than $39/25=585/375$, since $224^2\cdot30>1226^2$ and $84^2\cdot130>956^2$. Thus both inequalities in~\eqref{eq:scalar-conditions} hold with $\rho=39/25=1.56$, and Lemma~\ref{lem:scalar} proves the claimed bound. The polynomial running time was established in Section~\ref{sec:bkn}.
\end{proof}

\subsection{An improved guarantee via parameter search and exact verification}\label{sec:explicit-distribution}

We now allow $s,a$ to vary and to be correlated with $\gamma$, rather than keeping them fixed as in Section~\ref{sec:simple-guarantee}. We prove the stronger guarantee in Theorem~\ref{thm:main} by specifying a joint distribution $\eta$ satisfying~\eqref{eq:scalar-conditions} with $\rho:=1.55576$. The parameter choices are obtained by numerical search, and the resulting bound is verified using exact rational arithmetic.

For $y>0$, the inequality $1-yF(y)\le\rho(1-y)$ in~\eqref{eq:scalar-conditions} is equivalent to $F(y)\ge\rho-(\rho-1)/y$; thus it bounds $F$ from below. On the other hand, by~\eqref{eq:kernel}, $\gamma$ enters $\ell_{\gamma,s,a}(y)$ only through $\mathbf{1}_{\{\gamma\le y\}}$, whose coefficient $2-2\lambda\min\{y,a\}$ is nonnegative; hence increasing $\gamma$ while keeping $s,a$ unchanged does not increase $\ell_{\gamma,s,a}(y)$ or $L(y)$. We therefore take $\gamma$ as large as possible, that is, $F(y)=\Pr_\eta[\gamma\le y]$ as small as possible for every $y$. Since $F(y)\ge0$ and the lower bound $\rho-(\rho-1)/y$ is negative exactly for $y<(\rho-1)/\rho$, this gives
\begin{equation}\label{eq:gamma}
 F(y)=
 \begin{cases}
 0,&0\le y<(\rho-1)/\rho,\\
 \rho-(\rho-1)/y,&(\rho-1)/\rho\le y\le1.
 \end{cases}
\end{equation}
By construction,
\begin{equation}\label{eq:penalty}
 1-yF(y)\le\rho(1-y)\qquad\forall y\in[0,1].
\end{equation}

It remains to choose $s,a$ so that $L(y)\le\rho$ for every $y\in[0,1]$. Unlike the penalty bound, which depends only on the marginal distribution of $\gamma$, the length bound depends on the full joint distribution of $\gamma,s,a$.

We take $(s,a)$ to be a step function of $\gamma$. Let $\gamma_0=(\rho-1)/\rho$; since $F(\gamma_0)=0$, we have $\gamma>\gamma_0$ almost surely. Choose rational thresholds $\gamma_0<\gamma_1<\cdots<\gamma_{17}=1$ and, for each $i$, a parameter pair $(s_i,a_i)$ with $0\le s_i\le a_i\le1$ and $a_i>0$. Set $(s,a)=(s_i,a_i)$ when $\gamma$ lies in interval $i$, that is, $\gamma\in(\gamma_{i-1},\gamma_i]$. Table~\ref{tab:distribution} lists all thresholds and parameter pairs.

\begin{table}[htbp]
\centering
\footnotesize
\caption{Thresholds and parameter pairs: $(s,a)=(s_i,a_i)$ for $\gamma\in(\gamma_{i-1},\gamma_i]$, where $\gamma_0=(\rho-1)/\rho$ and $\rho=1.55576$.}\label{tab:distribution}
\begin{tabular}{@{}rlll@{\qquad}rlll@{\qquad}rlll@{}}
\toprule
$i$ & $\gamma_i$ & $s_i$ & $a_i$ & $i$ & $\gamma_i$ & $s_i$ & $a_i$ & $i$ & $\gamma_i$ & $s_i$ & $a_i$\\
\midrule
1 & 0.3582 & 0.8889 & 0.9111 & 7 & 0.7127 & 0.7222 & 1 & 13 & 0.8957 & 0.6994 & 1\\
2 & 0.3720 & 0.8400 & 0.9000 & 8 & 0.7214 & 0.7195 & 1 & 14 & 0.9133 & 0.6972 & 1\\
3 & 0.4023 & 0.8667 & 1 & 9 & 0.7300 & 0.7194 & 1 & 15 & 0.9662 & 0.6945 & 1\\
4 & 0.6852 & 0.7800 & 0.9800 & 10 & 0.7463 & 0.7111 & 1 & 16 & 0.9875 & 0.6917 & 1\\
5 & 0.6897 & 0.7250 & 1 & 11 & 0.8582 & 0.7093 & 1 & 17 & 1 & 0.6896 & 1\\
6 & 0.6971 & 0.7247 & 1 & 12 & 0.8850 & 0.7028 & 1 & & & & \\
\bottomrule
\end{tabular}
\end{table}

Interval $i$ has probability $F(\gamma_i)-F(\gamma_{i-1})$. Similarly, the event that $\gamma$ lies in interval $i$ and $\gamma\le y$ is $\{\gamma\le\min\{y,\gamma_i\}\}\setminus\{\gamma\le\min\{y,\gamma_{i-1}\}\}$, so its probability is $F(\min\{y,\gamma_i\})-F(\min\{y,\gamma_{i-1}\})$. Write $\lambda_i=1/(3a_i-s_i)$. When $\gamma$ lies in interval $i$, the terms $\lambda+2\lambda\max\{s-y,0\}$ of~\eqref{eq:kernel} equal $\lambda_i(1+2\max\{s_i-y,0\})$, and the remaining term $(2-2\lambda\min\{y,a\})\mathbf{1}_{\{\gamma\le y\}}$ equals $2-2\lambda_i\min\{y,a_i\}$ if $\gamma\le y$ and $0$ otherwise. Weighting these values by the two probabilities above and summing over $i$ gives
\begin{equation}\label{eq:explicit-psi}
\begin{aligned}
 L(y)={}&\sum_{i=1}^{17}\bigl(F(\gamma_i)-F(\gamma_{i-1})\bigr)\lambda_i\bigl(1+2\max\{s_i-y,0\}\bigr)\\
 &+\sum_{i=1}^{17}\bigl(F(\min\{y,\gamma_i\})-F(\min\{y,\gamma_{i-1}\})\bigr)\bigl(2-2\lambda_i\min\{y,a_i\}\bigr).
\end{aligned}
\end{equation}
Since all parameters are rational, \eqref{eq:explicit-psi} has a simple piecewise structure: its terms change their formulas only when $y$ passes one of the finitely many rational points $\gamma_j$, $s_i$, and $a_i$. For $y>0$, the inequality $L(y)\le\rho$ is equivalent to $y(\rho-L(y))\ge0$, and between two consecutive such points, $y(\rho-L(y))$ is a quadratic polynomial in $y$ with rational coefficients. Since a quadratic polynomial attains its minimum over an interval at an endpoint or at its critical point, this allows us to verify $L(y)\le\rho$ for every $y\in[0,1]$ by finitely many exact checks.

\begin{lemma}[Rational certificate]\label{lem:numerical}
For the thresholds and parameter pairs in Table~\ref{tab:distribution}, $L(y)\le\rho$ for every $y\in[0,1]$.
\end{lemma}

\begin{proof}
We split $[0,1]$ at $\gamma_0=(\rho-1)/\rho$.

For $y\le\gamma_0$, we have $F(y)=0$, so the second sum in~\eqref{eq:explicit-psi} vanishes and the first is nonincreasing in $y$. Hence it suffices to check $L(0)\le\rho$.

On $[\gamma_0,1]$, the terms in~\eqref{eq:explicit-psi} change their formulas only at the points
\[
 \{\gamma_j:0\le j\le17\}
 \cup\bigl(\{s_i,a_i:1\le i\le17\}\cap(\gamma_0,1)\bigr).
\]
For the values in Table~\ref{tab:distribution}, these are the $18$ points $\gamma_0<\gamma_1<\cdots<\gamma_{17}=1$, the $17$ values $s_i$, and the three values $a_i<1$. These $38$ points are distinct, so they divide $[\gamma_0,1]$ into $37$ intervals.

Let $[\underline{y},\overline{y}]$ be one of these intervals. Since no breakpoint lies in $(\underline{y},\overline{y})$, each term of~\eqref{eq:explicit-psi} is given by a single formula on $[\underline{y},\overline{y}]$. For each $i$, the term $\max\{s_i-y,0\}$ equals $s_i-y$ if $\overline{y}\le s_i$ and $0$ otherwise. Similarly, the term $\min\{y,a_i\}$ equals $y$ if $\overline{y}\le a_i$ and $a_i$ otherwise. The difference $F(\min\{y,\gamma_i\})-F(\min\{y,\gamma_{i-1}\})$ equals $0$ if $\overline{y}\le\gamma_{i-1}$ and $F(\gamma_i)-F(\gamma_{i-1})$ if $\underline{y}\ge\gamma_i$. Otherwise, $\gamma_{i-1}\le\underline{y}<\overline{y}\le\gamma_i$, and by~\eqref{eq:gamma} the difference equals $\rho-(\rho-1)/y-F(\gamma_{i-1})$. Expanding~\eqref{eq:explicit-psi} therefore gives rational numbers $q_0$, $q_1$, and $q_2$ such that
\[
 y\bigl(\rho-L(y)\bigr)=q_2y^2+q_1y+q_0\qquad\forall y\in[\underline{y},\overline{y}].
\]
Since a quadratic polynomial is determined by its values at three distinct points, these coefficients can be computed from the three pairs $\bigl(y,\,y(\rho-L(y))\bigr)$ with $y\in\{\underline{y},(\underline{y}+\overline{y})/2,\overline{y}\}$.

Since $y>0$, the inequality $L(y)\le\rho$ holds on $[\underline{y},\overline{y}]$ if and only if $q_2y^2+q_1y+q_0\ge0$ there. If $q_2\ne0$, this polynomial has the critical point $\hat y=-q_1/(2q_2)$, which is rational. Its minimum over $[\underline{y},\overline{y}]$ is attained at $\underline{y}$, at $\overline{y}$, or, if $q_2\ne0$ and $\underline{y}<\hat y<\overline{y}$, at $\hat y$. Hence it suffices to check $y(\rho-L(y))\ge0$ at $\underline{y}$, at $\overline{y}$, and, in the latter case, at $\hat y$.

The program in Appendix~\ref{app:certificate} performs all these checks in exact rational arithmetic: it checks $L(0)\le\rho$ and, on each of the $37$ intervals, checks $y(\rho-L(y))\ge0$ at $\underline{y}$, at $\overline{y}$, and at $\hat y$ whenever $q_2\ne0$ and $\underline{y}<\hat y<\overline{y}$; the last check is needed for $16$ intervals. All checks pass, which proves $L(y)\le\rho$ for every $y\in[0,1]$.
\end{proof}

\begin{proof}[Proof of Theorem~\ref{thm:main}]
By Lemma~\ref{lem:numerical} and~\eqref{eq:penalty}, both inequalities in~\eqref{eq:scalar-conditions} hold, so Lemma~\ref{lem:scalar} gives $\mathrm{ALG}\le\rho\,\mathrm{OPT}_{\mathrm{LP}}\le\rho\,\mathrm{OPT}$. The polynomial running time was established in Section~\ref{sec:bkn}.
\end{proof}

\begin{remark}[The limit of optimizing the parameter distribution]\label{rem:lower}
For every joint distribution $\eta$ allowed in Lemma~\ref{lem:scalar}, however $\gamma,s,a$ are correlated, the inequalities in~\eqref{eq:scalar-conditions} can hold only if
\begin{equation}\label{eq:lower-main}
 \rho\ge\frac{196969687}{126617839}\approx1.55562.
\end{equation}
The proof in Appendix~\ref{app:lower} is by contradiction. If some $\eta$ satisfied~\eqref{eq:scalar-conditions} with a smaller $\rho$, then, since the inequalities in~\eqref{eq:scalar-conditions} are required for every $y\in[0,1]$, $\eta$ would in particular satisfy them at $y=7/10$, $8/9$, and $1$; these alone already lead to a contradiction. Thus our factor $1.55576$ is within $1.4\cdot10^{-4}$ of the best factor that~\eqref{eq:scalar-conditions} can certify.
\end{remark}

\begin{remark}[Running-time comparison with BKN]\label{rem:shared-decomposition}
Omitting the splitting-off preprocessing simplifies Algorithm~\ref{alg:bkn}: it uses a single tree decomposition, whereas a direct implementation of the BKN algorithm computes a separate tree decomposition for each candidate splitting threshold. With a minor modification of the construction in BKN~\cite[Section~6, Lemma~5 and its proof]{BKN2026}, all these decompositions can be obtained from one sequence of complete splittings and its reversal.

Indeed, the construction in BKN introduces an auxiliary copy $r'$ of the root $r$, performs complete splittings at every vertex other than $r$ and $r'$, and then reverses these complete splittings to recover the tree decomposition. In the proof of~\cite[Lemma~5]{BKN2026}, the vertices are completely split in nondecreasing order of $y_v^*$, so the reversal restores them in the opposite order. Stopping the reversal once all vertices with $y_v^*\ge\delta$ have been restored gives the decomposition corresponding to the splitting threshold $\delta$. Thus the decomposition step need not be repeated for each splitting threshold, while the BKN algorithm still constructs tours from the decompositions obtained at the different thresholds.
\end{remark}

\section{Conclusion}\label{sec:conclusion}

We gave an LP-relative $1.55576$-approximation for PCTSP. The inequalities underlying our analysis cannot certify a factor below approximately $1.55562$ (Remark~\ref{rem:lower}). It remains open whether PCTSP admits an LP-relative $3/2$-approximation, matching the Christofides--Serdyukov guarantee for TSP~\cite{Christofides1976,Serdyukov1978}.

\paragraph{AI use.}
The author used OpenAI's ChatGPT and Anthropic's Claude to assist with text editing, the numerical searches for the parameters, and writing the exact verification code. The verification code is given in Appendix~\ref{app:certificate}; it checks the parameters in Table~\ref{tab:distribution} and the resulting bound in exact rational arithmetic and can be run independently. The author is responsible for all mathematical arguments, computations, and other content of the paper.

\setlength{\bibsep}{6pt plus 1pt minus 1pt}

\clearpage
\appendix
\section{Numerical certificate}\label{app:certificate}

The following Python~3 program verifies Lemma~\ref{lem:numerical} in exact rational arithmetic. It first checks the conditions on the thresholds and parameter pairs stated in Section~\ref{sec:explicit-distribution} and then performs the checks in the proof of Lemma~\ref{lem:numerical}. Here \texttt{table} contains the rows $(\gamma_i,s_i,a_i)$ of Table~\ref{tab:distribution}, each entry \texttt{(lo,hi,s,a,lam)} of \texttt{intervals} is $(\gamma_{i-1},\gamma_i,s_i,a_i,\lambda_i)$, \texttt{F(y)} and \texttt{L(y)} are computed by the formulas~\eqref{eq:gamma} and~\eqref{eq:explicit-psi}, and \texttt{ylo}, \texttt{yhi}, \texttt{q1}, \texttt{q2}, and \texttt{yhat} denote $\underline{y}$, $\overline{y}$, $q_1$, $q_2$, and $\hat y$.

\begin{lstlisting}
from fractions import Fraction as Frac

rho = Frac('1.55576')
gamma0 = (rho-1)/rho

def F(y):
    return max(Frac(0),rho-(rho-1)/y) if y else Frac(0)

table = [
    ('0.3582','0.8889','0.9111'),
    ('0.3720','0.8400','0.9000'),
    ('0.4023','0.8667','1'),
    ('0.6852','0.7800','0.9800'),
    ('0.6897','0.7250','1'),
    ('0.6971','0.7247','1'),
    ('0.7127','0.7222','1'),
    ('0.7214','0.7195','1'),
    ('0.7300','0.7194','1'),
    ('0.7463','0.7111','1'),
    ('0.8582','0.7093','1'),
    ('0.8850','0.7028','1'),
    ('0.8957','0.6994','1'),
    ('0.9133','0.6972','1'),
    ('0.9662','0.6945','1'),
    ('0.9875','0.6917','1'),
    ('1',     '0.6896','1'),
]

intervals = []
lo = gamma0
for row in table:
    hi,s,a = map(Frac,row)
    assert lo < hi and 0 <= s <= a <= 1 and a > 0
    intervals.append((lo,hi,s,a,1/(3*a-s)))
    lo = hi
assert lo == 1

def L(y):
    return sum((F(hi)-F(lo))*lam*(1+2*max(0,s-y))
               +(F(min(y,hi))-F(min(y,lo)))*(2-2*lam*min(y,a))
               for lo,hi,s,a,lam in intervals)

breaks = sorted({gamma0} | {hi for lo,hi,s,a,lam in intervals}
                | {y for lo,hi,s,a,lam in intervals for y in (s,a)
                   if gamma0 < y < 1})

assert L(0) <= rho
interior = 0
for ylo,yhi in zip(breaks,breaks[1:]):
    # y*(rho-L(y)) = q2*y^2+q1*y+q0 on [ylo,yhi]; use its values at ylo, m, yhi.
    m = (ylo+yhi)/2
    q2 = 2*(ylo*(rho-L(ylo))-2*m*(rho-L(m))+yhi*(rho-L(yhi)))/(yhi-ylo)**2
    q1 = (yhi*(rho-L(yhi))-ylo*(rho-L(ylo)))/(yhi-ylo)-2*q2*m
    assert ylo*(rho-L(ylo)) >= 0 and yhi*(rho-L(yhi)) >= 0
    if q2 != 0:
        yhat = -q1/(2*q2)
        if ylo < yhat < yhi:
            assert yhat*(rho-L(yhat)) >= 0
            interior += 1

print('PASSED:',len(breaks)-1,'intervals,',interior,'interior critical points.')
\end{lstlisting}
\clearpage

\section{A lower bound for the scalar analysis}\label{app:lower}

Suppose, for contradiction, that some joint distribution $\eta$ allowed in Lemma~\ref{lem:scalar}, with arbitrary correlations among $\gamma,s,a$, satisfies~\eqref{eq:scalar-conditions} for a constant $\rho<196969687/126617839$. The argument uses only the parameter ranges $0\le\gamma\le1$, $0\le s\le a\le1$, $a>0$ and the inequalities in~\eqref{eq:scalar-conditions} at three values of $y$, all of which hold for this $\eta$ because~\eqref{eq:scalar-conditions} holds at every $y$. We show that these inequalities force $\rho\ge196969687/126617839$.

The two inequalities in~\eqref{eq:scalar-conditions} constrain the distribution of $\gamma$ in opposite directions. As shown in Section~\ref{sec:explicit-distribution}, increasing $\gamma$ while keeping $s,a$ fixed does not increase $\ell_{\gamma,s,a}(y)$, so large values of $\gamma$ help to satisfy the first inequality, whereas the second inequality is equivalent to $F(y)\ge\rho-(\rho-1)/y$ for $y>0$. Thus $\gamma\le y$ must hold with probability at least $\rho-(\rho-1)/y$, so $\gamma$ cannot be large too often.

Let $(y_1,y_2,y_3)=(7/10,8/9,1)$ and $(\omega_1,\omega_2,\omega_3)=(9/50,277/1000,543/1000)$; these values were found by numerical search. Since $\omega_1+\omega_2+\omega_3=1$, the inequalities $L(y_i)\le\rho$ give $\sum_{i=1}^3\omega_iL(y_i)\le\rho$. As $L(y_i)=\mathbb{E}_\eta[\ell_{\gamma,s,a}(y_i)]$, we obtain a lower bound on $\sum_{i=1}^3\omega_iL(y_i)$ by bounding $\sum_{i=1}^3\omega_i\ell_{\gamma,s,a}(y_i)$ from below for every realization of $(\gamma,s,a)$ and then taking the expectation over $\eta$.

In this weighted sum, $\gamma$ appears only in the indicators $\mathbf{1}_{\{\gamma\le y_i\}}$. Since $y_1<y_2<y_3=1$, all three indicators equal one if $\gamma\le y_1$, only the last two if $y_1<\gamma\le y_2$, and only the last one if $y_2<\gamma\le1$. Hence the weighted sum depends on $\gamma$ only through which of the intervals $K_1=[0,y_1]$, $K_2=(y_1,y_2]$, and $K_3=(y_2,1]$ contains $\gamma$: if $\gamma\in K_j$, then with $\lambda=1/(3a-s)$, \eqref{eq:kernel} and $\sum_{i=1}^3\omega_i=1$ give
\[
 \sum_{i=1}^3\omega_i\ell_{\gamma,s,a}(y_i)
 =\lambda+2\sum_{i=j}^3\omega_i\bigl(1-\lambda\min\{y_i,a\}\bigr)+2\lambda\sum_{i=1}^3\omega_i\max\{s-y_i,0\}.
\]
The right-hand side depends only on $j$ and $(s,a)$; we denote it by $D_j(s,a)$.
We now determine constants $d_j$ such that $D_j(s,a)\ge d_j$ for all $0\le s\le a\le1$ with $a>0$. Since such a bound holds for all admissible pairs $(s,a)$, it holds however $s,a$ depend on $\gamma$. Combined with the probabilities of the intervals $K_j$, the constants $d_j$ therefore yield a lower bound on $\sum_{i=1}^3\omega_iL(y_i)$, and hence on $\rho$.

We take $d_j$ to be the minimum of $D_j(s,a)$ over nine pairs $(s,a)$ specified below. Since $3a-s\ge2a>0$ for $a>0$, the inequality $D_j(s,a)\ge d_j$ is equivalent to $(3a-s)\bigl(D_j(s,a)-d_j\bigr)\ge0$. Because $\lambda=1/(3a-s)$, expanding gives
\begin{equation}\label{eq:cleared}
\begin{split}
 (3a-s)\bigl(D_j(s,a)-d_j\bigr)
 &=1+\Bigl(2\sum_{i=j}^3\omega_i-d_j\Bigr)(3a-s)\\
 &\quad-2\sum_{i=j}^3\omega_i\min\{y_i,a\}+2\sum_{i=1}^3\omega_i\max\{s-y_i,0\}.
\end{split}
\end{equation}
The right-hand side of~\eqref{eq:cleared} is defined for all $0\le s\le a\le1$, including $s=a=0$, where it equals $1$. Consider a set of pairs $(s,a)$ with $s\le a$ in which $s$ and $a$ each range over one of the intervals $[0,y_1]$, $[y_1,y_2]$, and $[y_2,1]$. On such a set, every term $\min\{y_i,a\}$ equals either $y_i$ or $a$, and every term $\max\{s-y_i,0\}$ equals either $s-y_i$ or $0$. Hence the right-hand side of~\eqref{eq:cleared} is affine on each such set and attains its minimum over the set at an extreme point. The extreme points of these sets are the pairs $(s,a)$ with $s,a\in\{0,y_1,y_2,1\}$ and $s\le a$, that is, $s=a=0$ and the nine pairs
\[
 s\in\{0,y_1,y_2,1\},\qquad a\in\{y_1,y_2,1\},\qquad s\le a.
\]
The following table lists $D_1(s,a)$, $D_2(s,a)$, and $D_3(s,a)$ at these nine pairs $(s,a)$. For each $j$, the smallest displayed value of $D_j(s,a)$ is shown in bold and equals $d_j$. Thus $D_j(s,a)\ge d_j$ at the nine pairs; since $3a-s>0$ there, the right-hand side of~\eqref{eq:cleared} is nonnegative at these pairs, and it equals $1$ at $s=a=0$. As these are all the extreme points, the right-hand side of~\eqref{eq:cleared} is nonnegative for all $0\le s\le a\le1$, and~\eqref{eq:cleared} together with $3a-s>0$ gives $D_j(s,a)\ge d_j$ for all admissible $(s,a)$.

\begin{center}
\footnotesize
\setlength{\tabcolsep}{4pt}
\begin{tabular}{@{}ccccc@{\qquad}ccccc@{}}
\toprule
$s$ & $a$ & $D_1(s,a)$ & $D_2(s,a)$ & $D_3(s,a)$ & $s$ & $a$ & $D_1(s,a)$ & $D_2(s,a)$ & $D_3(s,a)$ \\
\midrule
$0$ & $7/10$ & $38/21$ & $824/525$ & $6301/5250$ & $0$ & $1$ & $23263/13500$ & $19537/13500$ & $793/750$ \\
$7/10$ & $7/10$ & $12/7$ & $537/350$ & $8801/7000$ & $7/10$ & $1$ & $16963/10350$ & $\boldsymbol{14371/10350}$ & $\boldsymbol{12059/11500}$ \\
$0$ & $8/9$ & $10403/6000$ & $881/600$ & $1099/1000$ & $8/9$ & $1$ & $\boldsymbol{15569/9500}$ & $13283/9500$ & $2559/2375$ \\
$7/10$ & $8/9$ & $7253/4425$ & $6227/4425$ & $32557/29500$ & $1$ & $1$ & $7513/4500$ & $323/225$ & $203/180$ \\
$8/9$ & $8/9$ & $1639/1000$ & $5683/4000$ & $183/160$ & & & & & \\
\bottomrule
\end{tabular}
\end{center}
The bold entries give
\[
 d_1=\frac{15569}{9500},\qquad
 d_2=\frac{14371}{10350},\qquad
 d_3=\frac{12059}{11500},
\]
and in particular $d_1>d_2>d_3$.

It remains to bound the probabilities of the three intervals. Since $0\le\gamma\le1$,
\[
 \Pr\nolimits_\eta[\gamma\in K_1]=F(y_1),\qquad
 \Pr\nolimits_\eta[\gamma\in K_2]=F(y_2)-F(y_1),\qquad
 \Pr\nolimits_\eta[\gamma\in K_3]=1-F(y_2).
\]
These probabilities are constrained by the penalty inequalities, which limit the probability of the later intervals, where $d_j$ is smaller. At $y=y_1$ and $y=y_2$, the penalty inequalities give $F(y_i)\ge\rho-(\rho-1)/y_i$, that is,
\[
 F(y_1)\ge1-\tfrac37(\rho-1),\qquad F(y_2)\ge1-\tfrac18(\rho-1).
\]
Therefore
\begin{align*}
 \rho&\ge\sum_{i=1}^3\omega_iL(y_i)
 =\sum_{j=1}^3\mathbb{E}_\eta\bigl[D_j(s,a)\,\mathbf{1}_{\{\gamma\in K_j\}}\bigr]
 \ge\sum_{j=1}^3d_j\Pr\nolimits_\eta[\gamma\in K_j]\\
 &=d_1F(y_1)+d_2\bigl(F(y_2)-F(y_1)\bigr)+d_3\bigl(1-F(y_2)\bigr)\\
 &=d_3+(d_2-d_3)F(y_2)+(d_1-d_2)F(y_1)\\
 &\ge d_3+(d_2-d_3)\bigl(1-\tfrac18(\rho-1)\bigr)+(d_1-d_2)\bigl(1-\tfrac37(\rho-1)\bigr)\\
 &=d_1-\bigl(\tfrac37(d_1-d_2)+\tfrac18(d_2-d_3)\bigr)(\rho-1)\\
 &=\frac{15569}{9500}-\frac{16493839}{110124000}(\rho-1).
\end{align*}
In the first line, the equality holds because $L(y_i)=\mathbb{E}_\eta[\ell_{\gamma,s,a}(y_i)]$ and $\sum_{i=1}^3\omega_i\ell_{\gamma,s,a}(y_i)=D_j(s,a)$ when $\gamma\in K_j$. The second inequality holds because $D_j(s,a)\mathbf{1}_{\{\gamma\in K_j\}}\ge d_j\mathbf{1}_{\{\gamma\in K_j\}}$ for every realization of $(\gamma,s,a)$. The inequality in the fourth line substitutes the bounds on $F(y_1)$ and $F(y_2)$, whose coefficients $d_2-d_3$ and $d_1-d_2$ are positive. Solving the resulting linear inequality for $\rho$ gives
\[
 \rho\ge\frac{196969687}{126617839}\approx1.55562,
\]
contradicting the assumption $\rho<196969687/126617839$. Hence~\eqref{eq:lower-main} holds.

\end{document}